\documentclass[conference]{IEEEtran}
\ifCLASSINFOpdf
  \usepackage[pdftex]{graphicx}
  \DeclareGraphicsExtensions{.pdf,.jpeg,.png}
\else
\fi
\usepackage[cmex10]{amsmath}
\usepackage{algorithmic}
\usepackage{url}

\usepackage{amsfonts}  

\begin{document}
%
\title{Computing Bounds on Network Capacity Regions as a Polytope Reconstruction Problem}

\author{\IEEEauthorblockN{Anthony Kim}
\IEEEauthorblockA{Oracle Corporation\\
500 Oracle Parkway, Redwood Shores, CA 94065\\
Email: tonyekim@yahoo.com}
\and
\IEEEauthorblockN{Muriel M\'{e}dard}
\IEEEauthorblockA{Research Laboratory of Electronics\\
MIT, Cambridge, MA 02139\\
Email: medard@mit.edu}}


%


\newtheorem{theorem}{Theorem}
\newtheorem{corollary}[theorem]{Corollary}
\newtheorem{lemma}[theorem]{Lemma}
\newtheorem{observation}[theorem]{Observation}
\newtheorem{proposition}[theorem]{Proposition}
\newtheorem{definition}[theorem]{Definition}
\newtheorem{claim}[theorem]{Claim}
\newtheorem{fact}[theorem]{Fact}
\newtheorem{assumption}[theorem]{Assumption}
\newtheorem{remark}[theorem]{Remark}
\newtheorem{example}[theorem]{Example}
\newtheorem{conjecture}[theorem]{Conjecture}
\numberwithin{equation}{section}   

\newcommand{\CC}{\mathbb C} 
\newcommand{\RR}{\mathbb R}
\newcommand{\ZZ}{\mathbb Z}
\newcommand{\QQ}{\mathbb Q}
\newcommand{\NN}{\mathbb N}
\newcommand{\FF}{\mathbb F}
\newcommand{\GGG}{\mathcal G}
\newcommand{\XXX}{\mathcal X}
\newcommand{\BBB}{\mathcal B}
\newcommand{\CCC}{\mathcal C}
\newcommand{\EEE}{\mathcal E}
\newcommand{\FFF}{\mathcal F}
\newcommand{\MMM}{\mathcal M}
\newcommand{\NNN}{\mathcal N}
\newcommand{\III}{\mathcal I}
\newcommand{\AAA}{\mathcal A}
\newcommand{\SSS}{\mathcal S}
\newcommand{\TTT}{\mathcal T}
\newcommand{\PPP}{\mathcal P}
\newcommand{\QQQ}{\mathcal Q}
\newcommand{\OOO}{\mathcal O}
\newcommand{\HHH}{\mathcal H}
\newcommand{\VVV}{\mathcal H}
\newcommand{\LLL}{\mathcal L}
\newcommand{\WWW}{\mathcal W}

\newcommand{\closure}{\operatorname{closure}}
\newcommand{\sign}{\operatorname{sign}}
\newcommand{\ideal}{\operatorname{Ideal}}
\newcommand{\In}{\operatorname{In}}
\newcommand{\Out}{\operatorname{Out}}
\newcommand{\Span}{\operatorname{span}}
\newcommand{\MinCost}{\operatorname{mincost}}
\newcommand{\vect}[1]{\boldsymbol{#1}}
\newcommand{\bd}{\operatorname{bd}}

\maketitle

\begin{abstract}
We define a notion of network capacity region of networks that generalizes the notion of network capacity defined by Cannons et al. and prove its notable properties such as closedness, boundedness and convexity when the finite field is fixed. We show that the network routing capacity region is a computable rational polytope and provide exact algorithms and approximation heuristics for computing the region. We define the semi-network linear coding capacity region, with respect to a fixed finite field, that inner bounds the corresponding network linear coding capacity region, show that it is a computable rational polytope, and provide exact algorithms and approximation heuristics. We show connections between computing these regions and a polytope reconstruction problem and some combinatorial optimization problems, such as the minimum cost directed Steiner tree problem. We provide an example to illustrate our results. The algorithms are not necessarily polynomial-time.
\end{abstract}


%
\IEEEpeerreviewmaketitle
\section{Introduction}
In a seminal work in 2000, Ahlswede et al.~\cite{ahlswede:network} introduced the network coding model to the problem of communicating information in networks. They showed that the extended capabilities of intermediate nodes to code on incoming packets give greater information throughput than in the traditional routing model and that the capacity of any multiple multicast network with a single source node is equal to the minimum of min-cuts between the source and receiver nodes.

A few explicit outer bounds of capacity regions of networks exist; the max-flow/min-cut bounds is one such set of bounds, which were sufficient in the case of single-source multiple multicast networks. Harvey et al.\cite{harvey:capacity} combined information theoretic and graph theoretic techniques to provide a computable outer bound on the network coding capacity regions. Yan et al.\cite{yan:outerbd} gave an explicit outer bound that improves upon the max-flow/min-cut bounds and showed its connection to a minimum cost network coding problem. They used their results to compute the capacity region of a special class of 3-layer networks. Thakor et al.\cite{thakor:capacity} gave a new computable outer bound, based on characterizations of all functional dependencies in networks. Yan et al.\cite{yan:capacity} provided an exact characterization of the capacity region for general multi-source multi-sink networks by bounding the constrained region in the entropy space. However, they noted that explicitly evaluating the obtained capacity regions remains difficult in general. In a related work, Chan and Grant\cite{chan:capacity} showed that even the explicit characterization of capacity region for single-source networks can be difficult since the computation of a capacity region involves the non-polyhedral set of entropy functions and that linear programming bounds do not suffice.

The routing capacity region of networks is better understood via linear programming approaches. Cannons et al.\cite{cannons:routing} defined a notion of network routing capacity that is computable with a linear program and showed that every rational number in $(0,1]$ is the routing capacity of some solvable network. Yazdi et al.\cite{yazdi:capacity1} and \cite{yazdi:capacity2} extended a special case of Farkas Lemma called the ``Japanese Theorem'' to reduce an infinite set of linear constraints to a finite set in terms of minimal Steiner trees and applied the results to obtain the routing capacity region of undirected ring networks. In a subsequent work, Kakhbod and Yazdi\cite{kakhbod:routing} provided complexity results on the description size of the inequalities obtained and applied them to the undirected ring networks. 

In this paper, we study the network capacity region of networks along the lines of work by Cannons et al.~\cite{cannons:routing}. We define the network capacity region of networks analogously to the rate regions in information theory and show its notable properties when the finite field is fixed: closedness, boundedness and convexity. In the case of routing, we prove that the network routing capacity region is a computable rational polytope and provide exact algorithms and approximation heuristics for computing the region. In the case of linear network coding, we define an auxiliary region, called the semi-network linear coding capacity region, which is a computable rational polytope that inner bounds the network linear coding capacity region, and provide exact algorithms and approximation heuristics, with respect to a fixed finite field. The main idea is to reduce computation of the capacity regions to a polytope reconstruction problem and use linear programming techniques on associated combinatorial optimization problems. Our results generalize to directed networks with cycles and undirected networks. This present work partially addresses two problems proposed by Cannons et al.~\cite{cannons:routing}: whether there exists efficient algorithms for computing their notions of network routing capacity and network linear coding capacity. It follows from our work that there exist combinatorial approximation algorithms, not polynomial-time, for computing the network routing capacity and for computing a lower bound of the network linear coding capacity. In addition, we provide an example to illustrate our results.

We note that algorithms and heuristics we present are not polynomial-time schemes. We do not distinguish the vertex and hyperplane descriptions of polytopes, but note that converting one description into another can be computationally expensive.

The paper is organized as follows. In Section II, we give a fractional network coding model. In Section III, we define the network coding capacity region and prove its notable properties. In Section IV and V, we prove notable properties of network coding capacity regions, or inner bounds thereof, in the case of routing and linear coding and discuss algorithmic aspects of the regions. In Section VI, we give an example that illustrates our results. Finally, in Section VII, we discuss further extensions and conclude. Due to the page limit, we omit some details and refer to \cite{kim:thesis} for more detailed exposition. 


\section{Fractional Network Coding Model}
We define a fractional network coding model. Most definitions are adapted from Cannons et al.~\cite{cannons:routing}. We write vectors or points in a multi-dimensional space with a hat as in $\hat{k}$ and let $\hat{k}_i$ (also $\hat{k}(i)$) denote the $i$-th coordinate of the vector. We shall omit the hat when convenient. 

A {\em capacitated network} $\NNN$ is a finite, directed, acyclic multigraph given by a 7-tuple $(\nu, \epsilon, \mu, c, \AAA, S, R)$ where $\nu$ is a node set, $\epsilon$ is an edge set, $\mu$ is a message set, $c:\epsilon \rightarrow \ZZ^+$ is an edge capacity function, $\AAA$ is an alphabet, $S:\nu \rightarrow 2^\mu$ is a source mapping, and $R:\nu \rightarrow 2^\mu$ is a receiver mapping. 

We define fractional edge function, fractional decoding function, and fractional message assignment with respect to a finite field $F$, where $|F| \geq |\AAA|$, a {\em source dimension vector} $\hat{k}$, and an {\em edge dimension} $n$. Let $\NNN = (\nu, \epsilon, \mu, c, \AAA, S, R)$ be a capacitated network and $m_1, \ldots, m_{\lvert \mu \rvert}$ be the messages. Let $\hat{k}=(k_1, \ldots, k_{\lvert \mu \rvert})$ be a vector of positive integers and $n$ be a positive integer. For each edge $e=(x,y)$, a {\em fractional edge function} is a map
$f_e:(F^{k_{i_1}})\times \cdots \times (F^{k_{i_\alpha}}) \times (F^{nc(e_{\alpha+1})})\times \cdots \times (F^{nc(e_{\alpha+\beta})})\rightarrow F^{nc(e)}$. For each node $x\in \nu$ and message $m_j\in R(x)$, a {\em fractional decoding function} is a map $f_{x,m_j}: (F^{k_{i_1}})\times \cdots \times (F^{k_{i_\alpha}}) \times (F^{nc(e_{\alpha+1})})\times \cdots \times (F^{nc(e_{\alpha+\beta})}) \rightarrow F^{k_j}$, where $m_{i_1}, \ldots, m_{i_\alpha}$ are $\alpha$ messages generated by $x$ and $e_{\alpha+1}, \ldots, e_{\alpha+\beta}$ are $\beta$ in-edges of $x$. We denote the collections of fractional edge and fractional decoding functions by $\FFF_e = \{f_e \: :\: e\in \epsilon\}$ and $\FFF_{x,m}=\{f_{x,m} \: :\: x\in \nu, m\in R(x)\}$. Note $k_i$ is the source dimension for message $m_i$.

A {\em fractional message assignment} is a collection of maps $\overline{a} = (a_1, \ldots, a_{\lvert \mu \rvert})$ where $a_i$ is a message assignment for $m_i$, $a_i:m_i \rightarrow F^{k_i}$. A {\em fractional network code} on $\NNN$ is a 5-tuple $(F, \hat{k}, n, \FFF_e, \FFF_d)$ where $F$ is a finite field, with $|F| \geq |\AAA|$, $\hat{k}$ is a source dimension vector, $n$ is an edge dimension, $\FFF_e$ is a collection of fractional edge functions, and $\FFF_d$ is a collection of fractional decoding functions. We have different classes of fractional network codes corresponding to $\FFF_e$ and $\FFF_d$ being routing, linear, and nonlinear functions over the field $F$: {\em fractional routing/linear/nonlinear network codes}. A fractional network code is a {\em fractional network code solution} if for every fractional message assignment $\overline{a}$, we have $f_{x,m_j}(a_{i_1}(m_{i_1}),\ldots, a_{i_\alpha}(m_{i_\alpha}), s(e_{\alpha+1}), \ldots, s(e_{\alpha+\beta}))$$ = a_j(m_j)$, for all $x\in \nu$ and $m_j\in R(x)$, where $m_{i_1}, \ldots, m_{i_\alpha}$ are $\alpha$ messages generated by $x$ and $e_{\alpha+1}, \ldots, e_{\alpha+\beta}$ are $\beta$ in-edges of $x$. If the above equation holds for a particular $x\in \nu$ and message $m\in R(x)$, then we say {\em node $x$'s demand $m$ is satisfied}. We have classes of fractional network code solutions: {\em fractional routing/linear/nonlinear network code solutions}. 

If $(F, \hat{k}, n, \FFF_e, \FFF_d)$ is a fractional network code solution for $\NNN$, source node $x\in \nu$ sends a vector of $k_i$ symbols from $F$, representing alphabets in $\AAA$, for each message $m_i \in S(x)$; each receiver node $x \in \nu$ demands the original vector of $k_i$ symbols corresponding to $m_i$ for each $m_i\in R(x)$; and each edge $e$ carries a vector of $c(e)n$ symbols. We refer to coordinates of the symbol vector corresponding to $m_i$ as {\em message $m_i$'s coordinates} and coordinates of the symbol vector on edge $e$ as {\em edge $e$'s coordinates}. Note that a coordinate of edge $e$ can be {\em active} or {\em inactive}, depending on whether it actively carries a symbol in the fractional network code solution or not. A fractional network code solution is {\em minimal} if the set $A$ of all active coordinates of edges is minimal, i.e., there exists no fractional network code solution with the same $F$, $\hat{k}$ and $n$ and the set of active coordinates that is a strict subset of $A$. 

In this paper, we consider non-degenerate networks; for each demand of a message at a node, there is a path from a source node generating the message to the receiver node and no message is both generated and demanded by the same node. We shall abridge network code descriptions and solution concepts when convenient. In this work, a multicast network has exactly one message produced by a source node and demanded by more than one receiver nodes and multi-source multi-sink networks are the most general network instances.

\section{Network Capacity Regions}


A vector of nonnegative numbers $\left(\frac{k_1}{n}, \ldots, \frac{k_{\lvert \mu \rvert}}{n}\right) \in \QQ_+^{\lvert \mu \rvert}$ is an {\em achievable coding rate vector} for $\NNN$ if there exists a fractional network code solution $(F, \hat{k}, n, \FFF_e, \FFF_d)$ where $\hat{k} = (k_1, \ldots, k_{\lvert \mu \rvert})$. The {\em network capacity region} of $\NNN$ is the closure of all achievable coding rate vectors in $\RR^{\lvert \mu \rvert}$. There are different classes of network capacity regions: the {\em network routing capacity region, $\CCC_r$,} which is the closure of all achievable routing rate vectors; the {\em network linear coding capacity region, $\CCC_l$,} which is the closure of all achievable linear coding rate vectors; and the {\em network coding capacity region, $\CCC$,} which is the closure of all achievable coding rate vectors.
\begin{theorem}
Let $\NNN$ be a capacitated network and $F$ a fixed finite field. The corresponding network capacity region $\CCC$ is a closed, bounded and convex set in $\RR_+^{|\mu|}$.
\end{theorem}
\begin{proof}
(Closedness) By definition, $\CCC$ is a closed set. (Boundedness) By symmetry, it suffices to show that $\frac{k_1}{n}$ is bounded in any achievable coding rate vector. Let $n$ be the edge dimension, $\nu_1$ be the set of source nodes in $\nu$ that generate message $m_1$ and $\gamma$ be the sum of capacities of out-edges of nodes in $\nu_1$. Then, $k_1 \leq \gamma n$ as we cannot send more than $\gamma n$ independent coordinates of message $m_1$ and expect receivers to recover all the information.  Hence, $\frac{k_1}{n} \leq \gamma$. (Convexity) Let $x_0, x_1 \in \CCC$ and $\lambda \in [0,1]$. It is straightforward to show that $x= (1-\lambda)x_0 + \lambda x_1 \in \CCC$ by rate sharing; we use sequences of achievable coding rate vectors converging to $x_0$ and $x_1$ to produce a sequence of achievable coding rate vectors converging to $x$.
\end{proof}

Similarly, the routing capacity region $\CCC_r$ and linear coding capacity region $\CCC_l$ are closed, bounded and convex regions when the finite field $F$ is fixed. However, the rate-sharing argument seems to break down in the case of linear coding if the finite field is not fixed. 


\section{Network Routing Capacity Regions}\label{sec:routing}
We show the network routing capacity region $\CCC_r$ is the image of a higher-dimensional rational polytope under an affine map and consider the computation of $\CCC_r$ as the polytope reconstruction problem with a ray oracle. Since multi-source multi-sink networks can be reduced to multiple multicast networks, it suffices to show the results with respect to the multiple multicast networks. 

\subsection{Properties}

\begin{theorem}\label{thm:routingregion}
The network routing capacity region $\CCC_r$ is a bounded rational polytope in $\RR_+^{|\mu|}$ and is computable.
\end{theorem}
\begin{proof}
(Polytope) It suffices to consider minimal fractional routing solutions, which consist of routing messages along Steiner trees, since any fractional routing solution can be reduced to a minimal one. Let $\TTT_i$ be the finite set of all Steiner trees rooted at the source node of message $m_i$ and spanning all receiver nodes that demand $m_i$, and $\TTT$ be the union, $\TTT = \TTT_1 \cup \ldots \cup \TTT_{|\mu|}$. Then, any minimal fractional routing solution satisfies the following constraints:
\[
\begin{array}{llll}
\sum_{T\in \TTT} T(e) \cdot x(T) & \leq & c(e)n, & \forall e\in \epsilon \\
\sum_{T\in \TTT_i} x(T)& = & k_i, & \forall 1\leq i \leq |\mu | \\
x & \geq & 0,
\end{array}
\]
where $x(T)$ is the number of times Steiner tree $T$ is used in the solution and $T(e)$ is an indicator that is 1 if $T$ uses edge $e$, or 0 otherwise. After scaling $x(T)$ by $n$, any minimal fractional routing solution satisfies
\[
\begin{array}{llll}
\sum_{T\in \TTT} T(e) \cdot x(T) & \leq & c(e), & \forall e\in \epsilon \\
x & \geq & 0. &
\end{array}
\]

As the coefficients are in $\QQ$, the above set of inequalities defines a bounded rational polytope $\PPP_r$, with rational extreme points. The polytope is bounded, because edge capacities are finite and no Steiner tree can be used to route for infinitely many times. Each minimal fractional routing solution reduces to a rational point inside the polytope $\PPP_r$, and each rational point $x$ inside $\PPP_r$ has a minimal fractional routing solution $(F, \hat{k}, n, \FFF_e, \FFF_d)$ that reduces to it, such that $\frac{k_i}{n} = \sum_{T\in\TTT_i}x(T)$ for all $i$. It follows that the network routing capacity region $\CCC_r$ is the image of $\PPP_r$ under the affine map $\psi_r: (x(T))_{T\in \TTT} \mapsto \left(\sum_{T\in\TTT_1} x(T), \ldots,\sum_{T\in\TTT_{|\mu|}}x(T)\right)$. It follows that $\CCC_r$ is a bounded rational polytope. 

(Computability) We first compute the vertices $v_1, \ldots, v_h$ of polytope $\PPP_r$ by any vertex enumeration algorithm. Then we compute the images of the vertices of $\PPP_r$ under the affine map $\psi_r$. The network routing capacity region is given by the vertices of the convex hull of points $\psi_r(v_1), \ldots, \psi_r(v_h)$ in $\RR_+^{|\mu|}$.
\end{proof}

\subsection{Algorithms}
We now provide exact algorithms and approximation heuristics for computing the network routing capacity region $\CCC_r$. We already provided an exact algorithm in the proof of Theorem~\ref{thm:routingregion}. Since the polytope $\PPP_r$ is defined in a high-dimensional space, the exact algorithm may not be efficient in practice. Using results in \cite{cole} and \cite{gritzmann} on polytope reconstruction problems, we derive different kinds of algorithms that might be more efficient. 

In our polytope reconstruction problem, we compute the facet description of a given polytope by making calls to the ray oracle $\OOO_{Ray}$ which, given a ray, returns the intersection point on the polytope surface and the ray. We reduce the computation of $\CCC_r$ to the polytope reconstruction problem by 1) reflecting $\CCC_r$ around the origin to get a symmetric polytope $\QQQ$ that contains the origin in its interior and 2) solving the linear programs similar to the one in Cannons et al.~\cite{cannons:routing} to implement the ray oracle $\OOO_{Ray}$. To reflect $\CCC_r$, we map all calls to the ray oracle to equivalent calls with rays defined in $\RR_+^{|\mu|}$. We use the algorithm outlined in Section 5 of Gritzmann et al.~\cite{gritzmann} to compute all the facets of the resulting polytope $\QQQ$ and, therefore, $\CCC_r$. The main idea of the algorithm is to first find a polytope $\QQQ'$ that contains $\QQQ$ and whose facet-defining hyperplanes are a subset of those for $\QQQ$ (Theorem 5.3 in \cite{gritzmann}), and then successively add more facet-defining hyperplanes of $\QQQ$ to $\QQQ'$ by using $\OOO_{Ray}$. By Theorem 5.5 in Gritzmann et al.~\cite{gritzmann} and the symmetries around the origin, we need at most $f_0(\CCC_r)+ (|\mu|-1) f_{|\mu|-1}^2(\CCC_r) + (5|\mu| - 4) f_{|\mu|-1}(\CCC_r)$ calls to the ray oracle where $f_i(\CCC_r)$ denotes the number of $i$-dimensional faces of $\CCC_r$ that do not contain the origin (the $0$-th dimensional faces being the points). 

If we use an exact algorithm for the ray oracle $\OOO_{Ray}$, we get an exact hyperplane description of the network routing capacity region. If instead we use an approximation algorithm that computes some point $r$ such that the actual intersection point lies between $r$ and $Ar$, then we obtain an approximation heuristic that computes a set of points $r$ such that the boundary of $\CCC_r$ lies between points $r$ and $Ar$. We note that an approximation algorithm for $\OOO_{Ray}$ does not necessary yield an approximation algorithm for $\CCC_r$, where an $A$-approximation of $\CCC_r$ would be a polytope $\PPP$ such that $\PPP \subset \CCC_r \subset A \PPP$. While an approximation algorithm for the oracle $\OOO_{Ray}$ does not necessarily lead to a polytope description of $\CCC_r$, it might be faster and more efficient than exact algorithms and, hence, more applicable to compute a quick ``sketch'' of $\CCC_r$. For instance, we can find approximate intersection points on a sufficiently large number of rays evenly spread apart.

\subsection{Implementations of Oracle $\OOO_{Ray}$}
Given the hyperplane description of the polytope $\PPP_r$ and a ray with a rational slope, $\hat{x} = \hat{q}t, t\geq 0$, we want to compute the rational intersection point of the ray and the boundary of $\CCC_r$. Note that the intersection point is exactly $\lambda_{max}\hat{q}$ where $\lambda_{max}$ is the optimal value to the linear program
\begin{equation}\label{lp:route}
\begin{array}{lllll}
\max & \lambda & & & \\
\operatorname{s.t.}& \sum_{T\in \TTT} T(e)\cdot x(T) & \leq & c(e), & \forall e\in \epsilon \\
 &\sum_{T\in \TTT_i} x(T) & \geq & \lambda q_i, & \forall i \\
 &x, \lambda & \geq &0. & \\
\end{array}
\end{equation}

We can use any linear programming algorithm, such as the ellipsoid algorithm, to solve the linear program exactly and, thus, obtain an exact oracle $\OOO_{Ray}$. For the approximate ray oracles, we can employ algorithms that solve the linear program within a provable approximation guarantee. Note that as the network routing capacity region $\CCC_r$ is a rational polytope, it suffices to consider rays with a rational slope.  

Using the results for the multicommodity flow and related problems by Garg and K\"{o}nemann~\cite{garg}, we can derive a combinatorial approximation algorithm that solves the linear program~\eqref{lp:route}. It computes a point $\hat{r}$ such that $\lambda_{max}\hat{q}$ is on the line segment between $\hat{r}$ and $(1+\omega)A\hat{r}$ for some numbers $\omega>0$ and $A\geq 1$. The main idea is to view solving \eqref{lp:route} as concurrently packing Steiner trees according to the ratio defined by $\hat{q}$ and use the results for the minimum cost directed Steiner tree problem: given an acyclic directed multigraph $\GGG = (\nu, \epsilon)$, a length function $l:\epsilon \rightarrow \RR^+$, a source node $s$ and receiver nodes $n_1, \ldots, n_k$, find a minimum cost (defined by $\sum_{e\in \epsilon'} l(e)$) subset of edges $\epsilon'$ such that there is a directed path from $s$ to each $n_i$ in $\epsilon'$. We assume we have an oracle $\OOO_{DSteiner}$ that solves the minimum cost directed Steiner tree problem, which is well-known to be NP-hard, within an approximation guarantee $A$. Then we have the following theorem:

\begin{theorem}\label{routing:step2}
There exists an $(1+\omega)A$-approx\-imation algorithm for the linear program~\eqref{lp:route} in time $O(\omega^{-2}$ $ (|\mu|\log A|\mu|$ $+ |\epsilon|)A\log |\epsilon| \cdot T_{DSteiner})$, where $T_{DSteiner}$ is the time required to solve the minimum cost directed Steiner tree problem with oracle $\OOO_{DSteiner}$ within an approximation guarantee $A$.
\end{theorem}

Note that computing the network routing capacity region is not the same as the multicommodity flow, as in our problem, a message can be demanded by multiple nodes and can be duplicated at intermediate nodes. There are algorithms for $\OOO_{DSteiner}$; the approximation using shortest paths yields $A=O(|\nu|)$ in time $O((|\nu|^2 + |\nu||\epsilon|) \log |\nu|)$ and Charikar et al.~\cite{charikar} gives a set of algorithms with $A=i(i-1)|\nu|^{1/i}$ and time $O(|\nu|^{3i})$ for any integer $i>1$. We can also use any brute-force algorithm that yields $A=1$. Note that further improvements in the minimum cost directed Steiner tree problem translates to improvements to our approximation algorithm by above theorem. 

\section{Network Linear Coding Capacity Regions}
We define  a computable polytope $\CCC_l'$, which we call the {\em semi-network linear coding capacity region}, that is contained in the network linear coding capacity region $\CCC_l$. It is unknown how good of an approximation the polytope $\CCC_l'$ is to $\CCC_l$. Unlike in the last section, the finite field is important in the computation of $\CCC_l'$ because of linearity of functions and termination of algorithms. In this section, we assume that a network $\NNN$ is a multi-source multi-sink network and that the finite field $F$ is fixed.

\subsection{Definitions}
The {\em weight vectors associated with $\NNN$} are vectors $w$ in $\{0,1\}^{|\mu|}$ such that there exists a scalar-linear network code solution when only messages $m_i$ with $w_i=1$ are considered and all the edges are assumed to be of unit capacity. We refer to the scalar-linear network code solutions corresponding to these weight vectors as {\em partial scalar-linear network code solutions}. A fractional network code is a {\em simple fractional linear network code solution} if the fractional network code is linear over the finite field $F$ and can be written as an aggregate of partial scalar-linear solutions (when considered with unit edge capacities). We define the {\em semi-network linear coding capacity region} $\CCC'_l$ of $\NNN$ as the closure of all coding rate vectors achievable by simple fractional linear network code solutions. Clearly, the network linear coding capacity region $\CCC_l$ contains the semi-network linear coding capacity region $\CCC'_l$. Note that scalar-linear network code solution and simple fractional linear network code solution are two different solution concepts.


\subsection{Results}
By the same line of reasoning as in the case of routing, we get the following theorems. 

\begin{theorem}
Assume a finite field $F$ is given. The semi-network linear coding capacity region $\CCC'_l$, with respect to $F$, is a bounded rational polytope in $\RR_+^{|\mu|}$ and is computable.
\end{theorem}

Instead of the minimum cost directed Steiner tree problem, we have two associated problems: minimum cost scalar-linear network code problem where given a network with unit edge capacities, a finite field $F$, and a length function $l:\epsilon\rightarrow \RR_+$, we want to compute the minimum cost (defined to be the sum of lengths of the edges used) scalar-linear network code solution; and fractional covering with box constraints problem which is to solve 
\begin{equation*}
\begin{array}{lllll}
\min & \sum_{j=1}^m c(j)x(j) & & & \\
\operatorname{s.t.}& \sum_j A(i,j) x(j) & \geq & b(i) , & \forall 1 \leq i \leq n \\
 & x(j) & \leq & u(j), & \forall 1\leq j \leq m \\
 &x & \geq &0, & 
\end{array}
\end{equation*}
where $A$ is an $n\times m$ nonnegative integer matrix, $b$ a nonnegative vector, $c$ a positive vector, and $u$ a nonnegative integer vector. Assuming we have two oracles, $\OOO_{SLinear}$ and $T_{FCover}$, for these problems, we get:

\begin{theorem}\label{lcoding:step2}
There exists an $(1+\omega)B$-approx\-imation algorithm for oracle $\OOO_{Ray}$, for $\CCC'_l$, with time $O(\omega^{-2} (\log A|\mu| + |\epsilon|) B \log |\epsilon| \cdot (T_{FCover} + k' T_{SLinear}))$, where $T_{FCover}$ is the time required to solve the fractional covering problem by $\OOO_{FCover}$ within an approximation guarantee $B$, $T_{SLinear}$ is the time required to solve the minimum cost scalar-linear network code problem exactly by $\OOO_{SLinear}$, and $k'$ is the total number of weight vectors associated with $\NNN$.
\end{theorem}

For $\OOO_{FCover}$, we can use any linear programming algorithm or the combinatorial approximation algorithm by Fleischer~\cite{fleischer}. For $\OOO_{SLinear}$, there is a dynamic programming algorithm in \cite{kim:thesis}. Without the minimum cost condition, the scalar-linear network code problem reduces to the decidability problem of determining whether or not a network has a scalar-linear solution, which is NP-hard~\cite{lehman}. Without the fixed finite field $F$, the decidability of the problem is unknown; this justifies in part our assumption of the fixed finite field $F$.

\section{An Example}
For network $\NNN$, we computed the network routing capacity region (two inner curves) and the semi-network linear coding capacity region (two outer curves) with both exact and approximate ray oracles in Figure~\ref{fig:graphic}. Network $\NNN$ has two source nodes at the top and two receiver nodes at the bottom. For exact ray oracles, we hard-coded the corresponding linear programs and used a linear program solver, linprog, in MATLAB. To obtain approximate intersections points, we simply used an approximate oracle $\OOO_{Ray}$ in place of the linear program solver in an implementation of a 2D-polytope reconstruction algorithm. We note that the approximate oracle $\OOO_{Ray}$ worked well in this example and led to its successful termination, but this may not hold for arbitrary networks in general.

\begin{figure}[!t]
\centering
\includegraphics[width=3.4in]{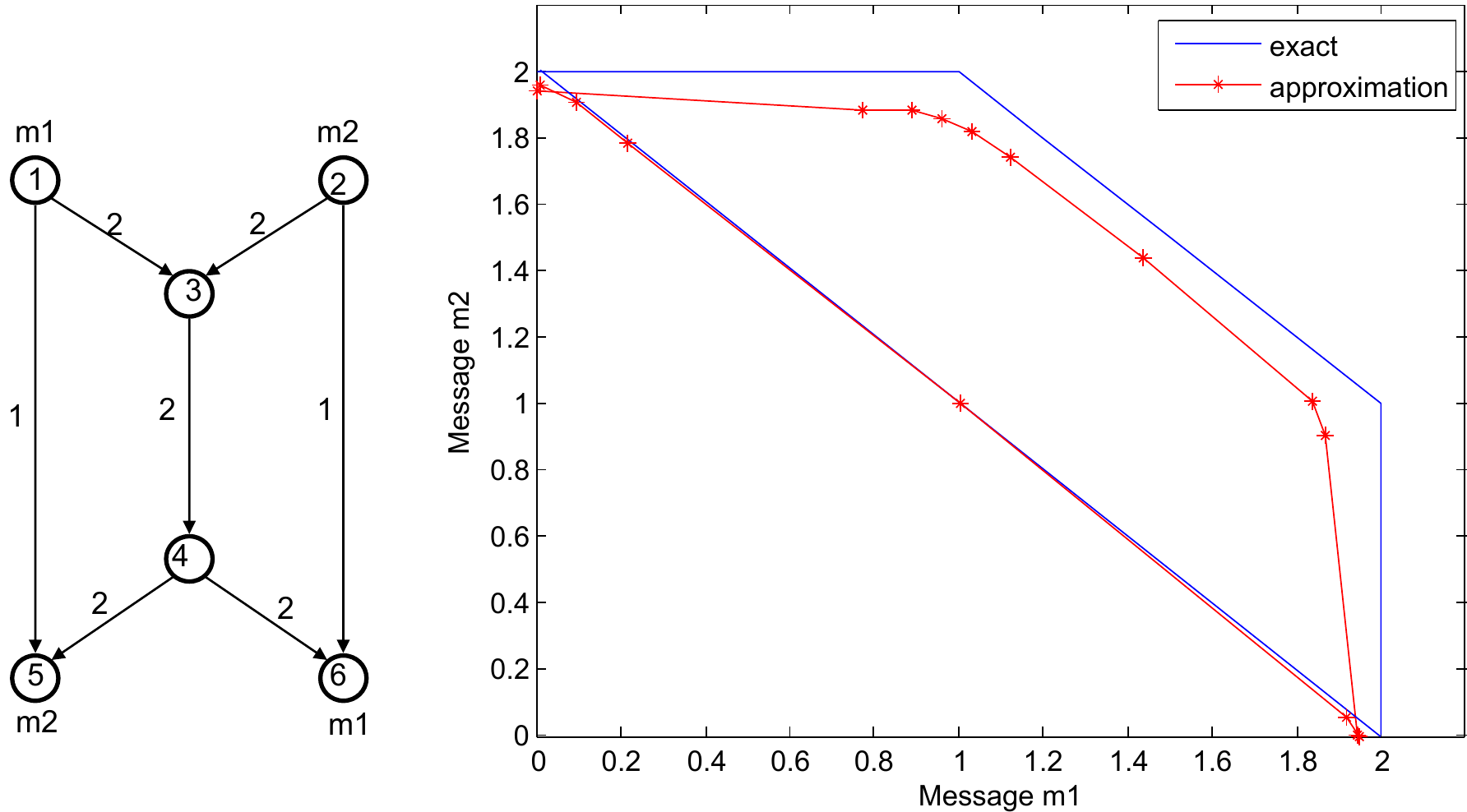}
\caption{Network $\NNN$ and its network routing capacity region and semi-network linear capacity region, each computed with an exact $\OOO_{Ray}$ and an approximate $\OOO_{Ray}$.}
\label{fig:graphic}
\end{figure}

\section{Further Discussion}
Our results have a few extensions: We can design membership testing algorithms that given a rate vector, determines whether or not there exists a fractional network code solution that achieves the rates from algorithms we have provided, for the network routing capacity region and semi-network linear coding capacity region. We can also generalize the results to directed networks with cycles and undirected networks straightforwardly. 

Note that the network routing capacity defined by Cannons et al.\cite{cannons:routing} corresponds to a point on the boundary of polytope $\CCC_r$; it is exactly the intersection point between the (outer) boundary of $\CCC_r$ and the ray $\hat{x} = (1, \ldots, 1) t, t\geq 0$. Hence, our work partially addresses two problems proposed by Cannons et al.: whether there exists efficient algorithms for computing their notions of network routing capacity and network linear coding capacity. It follows from our work that there exist combinatorial approximation algorithms, albeit not polynomial-time, for computing the network routing capacity and for computing a lower bound of the network linear coding capacity.

For details omitted in this paper, we refer to \cite{kim:thesis}.


\section*{Acknowledgment}
The authors would like to thank Michel Goemans and the anonymous reviewers for helpful comments. This material is based upon work supported by the Air Force Office of  Scientific Research (AFOSR) under award No. 016974-002.



\bibliographystyle{IEEEtran}
\bibliography{IEEEabrv,mybib}
%

%
%

\end{document}